\documentclass{scrartcl}
\usepackage[colorlinks=true]{hyperref}

\usepackage{graphicx}
\usepackage{amssymb}
\usepackage{amsmath}
\usepackage{amsthm}
\usepackage{xcolor}
\usepackage{algpseudocode}
\usepackage{algorithm}
\usepackage{mathtools}
\usepackage{float}
\usepackage{kotex}
\usepackage{microtype}
\usepackage{comment}

\usepackage{thmtools}
\usepackage[capitalize]{cleveref} 

\usepackage{algorithm,float,tabularx,amsfonts}

\crefname{figure}{Figure}{figures}

\newcommand{\R}{\mathbb{R}}
\newcommand{\N}{\mathbb{N}}

\makeatletter
\newcommand{\multiline}[1]{%
    \begin{tabularx}{\dimexpr\linewidth-\ALG@thistlm}[t]{@{}X@{}}
        #1
    \end{tabularx}
}

\makeatletter
\DeclareFontFamily{U}{tipa}{}
\DeclareFontShape{U}{tipa}{m}{n}{<->tipa10}{}
\newcommand{\arc@char}{{\usefont{U}{tipa}{m}{n}\symbol{62}}}%

\newcommand{\arc}[1]{\mathpalette\arc@arc{#1}}

\newcommand{\arc@arc}[2]{%
  \sbox0{$\m@th#1#2$}%
  \vbox{
    \hbox{\resizebox{\wd0}{\height}{\arc@char}}
    \nointerlineskip
    \box0
  }%
}
\makeatother

\newcommand{\partitions}{\textsf{S}}

\DeclareMathOperator{\tsp}{\mathsf{TSP}}
\newcommand{\TSP}{\tsp} 

\newtheorem{lemma}{Lemma}

\title{Balanced TSP partitioning}

\author{Benjamin Aram Berendsohn\thanks{Max Planck Institute for Informatics, Saarbrücken, Germany. Work supported by DFG grant KO 6140/1-2. \texttt{benjamin.berendsohn@fu-berlin.de}}
\and
Hwi Kim\thanks{Department of Computer Science and Engineering, Pohang University of Science and Technology, Republic of Korea. \texttt{hwikim@postech.ac.kr}.}
\and
László Kozma\thanks{Institut für Informatik, Freie Universität Berlin, Germany. Work supported by DFG grant KO 6140/1-2. \texttt{laszlo.kozma@fu-berlin.de}.}}

\date{}

\begin{document}
\maketitle

\begin{abstract}
The \emph{traveling salesman problem} (TSP) famously asks for a shortest tour that a salesperson can take to visit a given set of cities in any order. In this paper, we ask how much faster $k \ge 2$ salespeople can visit the cities if they divide the task among themselves. We show that, in the two-dimensional Euclidean setting, two salespeople can always achieve a speedup of at least $\frac12 + \frac1\pi \approx 0.818$, for any given input, and there are inputs where they cannot do better. We also give (non-matching) upper and lower bounds for $k \geq 3$.
\end{abstract}

\section{Introduction}

The \emph{traveling salesman problem} (TSP) asks, given $n$ cities and their pairwise distances, for the shortest tour that visits each city.
It is one of the best studied optimization problems, well-known to be NP-hard even in the planar Euclidean case, i.e., if the cities are points in $\mathbb{R}^2$, and the distance between any two points is the Euclidean distance~\cite{GareyJohnson}. This restricted version of the problem admits a polynomial-time approximation scheme (PTAS)~\cite{Arora, Mitchell}, whereas the more general case of metric distances is known to admit an approximation ratio slightly below $1.5$ (from the recent breakthrough~\cite{tsp_recent} that improved the longstanding approximation ratio of $1.5$ due to Christofides~\cite{Christofides}). 

In this paper, we focus on the two-dimensional Euclidean case and study the following natural question: How much faster can all cities be visited if multiple salespeople can collaborate on performing the task, and each city has to be visited by at least one salesperson? More formally, we wish to cover a given point set $P$ with multiple (say,~$k$) closed curves, instead of a single one, minimizing the maximum length among the $k$ curves. Intuitively, the $k$ salespeople execute their tours simultaneously, all have to return to their respective (arbitrary) starting points, and the goal is to have all of them finish before a given deadline.

The specific question we ask is: How does the cost improve with the parameter $k$, when compared to the normal TSP cost, in the worst case?
This ratio, which we precisely define next, can be seen as an inherent measure of \emph{decomposability} of the TSP problem.

A \emph{tour} of a point set $P \subset \R^2$ is a closed polygonal curve that contains each point in $P$.
A tour is \emph{optimal} if its length is minimal among all tours.\footnote{We do not require the vertices of a tour to be from $P$, though this is clearly the case for \emph{optimal} tours.} The length of an optimal tour of $P$ is denoted by $\TSP(P)$. Note that the optimal tour is necessarily a simple polygon, and in particular, it visits each point in $P$ exactly once. (We still allow degenerate cases, such as polygons with one or two corners, or polygons on multiple collinear points.)
Let $\partitions_k(P)$ denote the set of partitions of $P$ into at most $k$ subsets, and let
\[ \TSP_k(P) = \min_{R \in \partitions_k(P)} \max_{Q \in R} \TSP(Q). \]
Intuitively, the quantity $\TSP_k(P)$ corresponds to the least amount of time $k$ salespeople need to serve the points in $P$. In particular, $\TSP_1(P) = \TSP(P)$.

The ratio $\gamma(P,k) = \TSP_k(P) / \TSP(P)$ indicates the advantage that $k$ salespeople have over a single one. Observe that always $\gamma(P,k) \le 1$, but $\gamma(P,k)$ can be arbitrarily small (e.g., if $P$ consists of two very small clusters that are far apart from each other). In consequence, it makes sense to ask how \emph{large} $\gamma(P,k)$ can get, i.e., how much of an improvement we can \emph{guarantee} when using multiple salespeople. Accordingly, we define:
\[ \gamma(k) = \sup_P \gamma(P,k) = \sup_P 
\left( \frac{\TSP_k(P)}{\TSP(P)} 
\right). \]

Clearly, $\gamma(1) = 1$.
The main result of this paper is the precise determination of $\gamma(2)$.

\begin{restatable}{theorem}{restateBoundsTwo}\label{thm:bounds-2}
	$\gamma(2) = \tfrac12 + \tfrac1\pi \approx 0.818$.
\end{restatable}

We further give some lower and upper bounds for $\gamma(k)$ when $k \ge 3$.

\begin{restatable}{theorem}{restateLBGeneral}\label{thm:lb-general}
	For all $k \ge 2$, we have $\gamma(k) \ge \frac{1}{k} + \frac{1}{\pi} \sin \frac{\pi}{k}$.
\end{restatable}


\begin{restatable}{theorem}{restateUBMultiply}\label{thm:ub-multiply}
	For all $a, b \in \N$, we have $\gamma(a \cdot b) \le \gamma(a) \cdot \gamma(b)$.
\end{restatable}

\begin{restatable}{theorem}{restateSplitUB}\label{thm:ub-split}
	For all $a, b \in \N$, we have \[ \gamma(a+b) \le \left( 1 + \tfrac 2 \pi \right) \cdot \frac{\gamma(a) \cdot \gamma(b)}{\gamma(a)+\gamma(b)}. \]
\end{restatable}

We leave determining tight bounds for $k \geq 3$ as a challenging open question. Approximate results for small values of $k$ are found in \cref{tab:overview}. A generalization of the problem to higher dimensions or to more general metric spaces is likewise interesting. 

%
%

\begin{table}
	\centering
	\begin{tabular}{rrrl}
		$k$ & Lower bound & Upper bound & Upper bound reference\\
		\hline
		1 & 1 & 1 & trivial\\
		2 & $\approx 0.818$ & $\approx 0.818$ & Thm.~\ref{thm:bounds-2}\\
		3 & $\approx 0.609$ & $\approx 0.737$ & Thm.~\ref{thm:ub-split} with $1+2=3$\\
		4 & $\approx 0.475$ & $\approx 0.670$ & Thm.~\ref{thm:ub-multiply} with $2\cdot2=4$\\
		5 & $\approx 0.387$ & $\approx 0.634$ & Thm.~\ref{thm:ub-split} with $2+3=5$\\
		6 & $\approx 0.326$ & $\approx 0.603$ & Thm.~\ref{thm:ub-multiply} with $2\cdot3=6$\\
		7 & $\approx 0.281$ & $\approx 0.574$ & Thm.~\ref{thm:ub-split} with $3+4=7$\\
		8 & $\approx 0.247$ & $\approx 0.548$ & Thm.~\ref{thm:ub-multiply} with $2\cdot4=8$\\
		9 & $\approx 0.220$ & $\approx 0.533$ & Thm.~\ref{thm:ub-split} with $4+5=9$\\
		10 & $\approx 0.198$ & $\approx 0.519$ & Thm.~\ref{thm:ub-multiply} with $2\cdot5=10$
	\end{tabular}
	\caption{Approximate bounds for $\gamma(k)$, $1 \le k \le 10$. All lower bounds are derived from Theorem~\ref{thm:lb-general}.}\label{tab:overview}
\end{table}

\subsection{Related work} Although many variants of TSP with multiple salespeople have been studied in the literature, 
we are not aware of the ratio $\gamma(k)$ being explicitly considered before. 
We mention a few works that study problems of a similar flavor. 

The \emph{min-max cycle cover} problem refers to a quantity essentially equivalent to $\TSP_k$ in a more general weighted graph setting; e.g., see~\cite{yu2016improved}. In the TSP literature, closely related problems include the \emph{$k$-person TSP},  \emph{multiple TSP}, and \emph{multi-depot vehicle routing}, e.g., see~\cite{bektas2006multiple, xu20103, montoya2015literature, braekers2016vehicle} (and references therein). These problems are studied under various optimization criteria and often with additional constraints; for instance, in several formulations a starting point for each salesperson, or alternatively, a common starting point for all, are specified. A ``depot-free'' variant essentially matching our setting has been studied in~\cite{cornejo2023compact}. All these works aim at computing optimal or approximate solutions for given input instances, whereas our concern is the worst-case ratio $\gamma(k)$ in a geometric/Euclidean case. 

In a geometric setting, a cost-ratio similar to ours has been studied for the \emph{minimum spanning tree} (MST) problem, but with \emph{sum}, rather than \emph{maximum} criterion, e.g., see~\cite{ameli, dumitrescu}. Bereg et al.~\cite{bereg} study a partitioned MST and TSP problem, where starting points are given and the parts must have \emph{equal cardinality}. Arkin et al.~\cite{arkin} and Johnson~\cite{johnson} study optimization problems under a two-partitioning similar to ours, but with a pairing of points given as part of the input, requiring each pair to be separated.

\section{Preliminaries}

A \emph{curve} $C$ is the image of a continuous function $c \colon [0,\ell] \rightarrow \R^2$, for some $\ell \in \R_+$. We call $c$ a \emph{parametrization} of $C$.
If $c(0) = c(\ell)$, then the curve is \emph{closed}.
It is sometimes useful to extend the domain of $c$ to $\R_{\ge 0}$; in this case, we set $c(x) = c(x) \bmod \ell$ if $x > \ell$. Observe that this function is still continuous for closed curves.


We define \emph{polygons} and \emph{polygonal curves} in the usual way. Observe that a polygonal curve $C$ can be parametrized with a function $c \colon [0,\ell] \rightarrow \R^2$ such that for two points $c(x)$ and $c(y)$ on the curve, the distance in clockwise direction along the curve between $c(x)$ and $c(y)$ is precisely $|x-y|$. We then call $c$ a \emph{canonical parametrization} of $C$. We write $|C| = \ell$ for the length of $C$.

Given a closed curve $C$ and two distinct points $p, q$ on $C$, we call the line segment $pq$ a \emph{diagonal}, and we denote by $C(p, q)$ the subpath of $C$ from $p$ to $q$ in the positive direction according to the parametrization. Formally, let $c \colon [0,\ell] \rightarrow \R^2$ be a parametrization of $C$ and let $p = c(x_1)$ and $q = c(x_2)$. Then $C(p, q)$ is the curve parametrized by $c'$, where $c'$ is $c$ restricted to the interval $[x_1, x_2]$ if $x_1 < x_2$, or to the interval $[x_2, \ell + x_1]$ otherwise.

\begin{figure}[ht]
	\centering
	\includegraphics{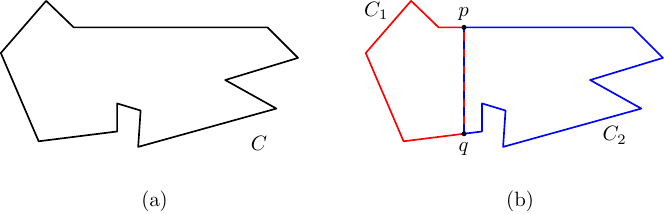}
	\caption{(a) A curve $C$. (b) Curves $C(p,q)$ and $C(q,p)$, assuming a clockwise parametrization.}
	\label{fig:inducedCurves}
\end{figure}

\subparagraph{Width.}

Given an angle $\theta$, let $u_\theta = (\cos \theta, \sin \theta)$ be the unit vector in direction $\theta$.
The \emph{(directional) width} of a closed curve $C$ in direction $\theta$ is the minimum distance between two parallel lines, orthogonal to $u_\theta$, that enclose $C$ (see \cref{fig:width}). Equivalently, it is the length of the projection of $C$ to a line parallel to $u_\theta$. Formally, we write:
\[ w(C, u_{\theta}) = \max_{p \in C} \langle u_{\theta}, p \rangle  - \min_{p \in C} \langle u_{\theta}, p \rangle. \]
The \emph{width} of a closed curve $C$ is the minimum width over all directions:
\[ w(C) = \min_{\theta \in [0, \pi)} w(C,u_{\theta}). \]

\begin{figure}[t]
	\begin{center}
		\includegraphics[width=.7\textwidth]{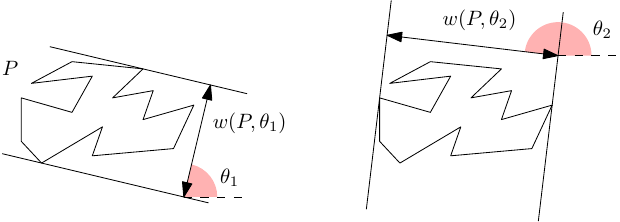}
	\end{center}
	\caption{A polygon $P$ and its width with respect to directions $\theta_1$ and $\theta_2$. Here $w(P, \theta_1) < w(P, \theta_2)$.}
	\label{fig:width}
\end{figure}

It is well known~\cite{cauchy} that the \emph{mean width} (computed by integration over all directions) of a closed convex curve is precisely its length divided by $\pi$. Since the width of a curve is clearly at most its mean width, we have:

\begin{lemma}\label{lem:cauchy}
	For any convex closed curve $C$ on $\mathbb{R}^2$, 
	$w(C) \le |C|/\pi$.
\end{lemma}

\section{Lower bounds: the circular point set}\label{sec:circle}

We start with an illustrative special case that will also provide our lower bound for the case $k = 2$. Let $P_n$ be a set of $n \ge 2$ points, arranged regularly-spaced along the unit circle. More precisely, define $P_n = \{ p_i \mid 1 \le i \le n \}$ with $p_i = ( \cos (2\pi \tfrac i n), \sin (2\pi \tfrac i n))$. See \cref{fig:evenPointsProof}~(a).

\begin{figure}
    \centering
    \includegraphics{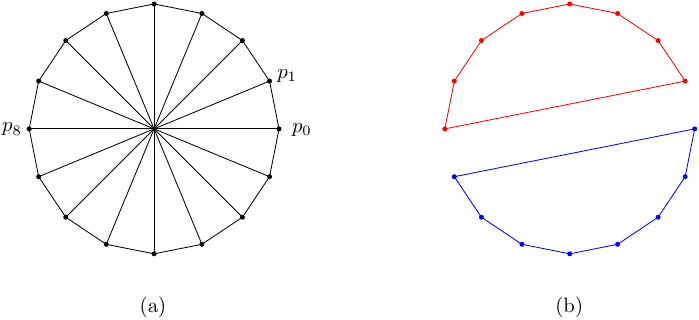}
    \caption{(a) Point set $P_{16}$. (b) Balanced partition of $P_{16}$.}
    \label{fig:evenPointsProof}
\end{figure}

Clearly, the shortest tour of $P_n$ goes around the circle, approaching a length of $2 \pi$ as $n$ goes to infinity. Thus, we have $\TSP(P_n) \rightarrow 2 \pi$. Now, consider the case $k = 2$. The intuitively best way of splitting the point set is shown in \cref{fig:evenPointsProof}~(b): Divide the circle with a straight cut into equal parts (assume for simplicity that $n$ is even). For both parts, take the half-circle plus the diagonal as a tour. As $n$ goes to infinity, the length of each of these tours will tend towards $\pi + 2$ (half the circumference plus twice the radius). Thus, the following holds:
\[ \lim_{n \rightarrow \infty} \gamma(P_n,2) \le \tfrac{\pi+2}{2 \pi} = \tfrac12 + \tfrac1\pi. \]

Later, in \cref{sec:ub}, we show that this upper bound holds for \emph{every} point set, using a similar technique of ``halving'' the optimal tour.

For our circular point set, it turns out that the bound is tight, since the above construction is optimal, as we show now. The main technical lemma is the following.

\begin{lemma}\label{lem:circle-subsets}
	Let $P_{2n} = \{ p_i \mid 1 \le i \le 2n \}$ with $p_i = ( \cos (2\pi \tfrac i {2n}), \sin (2\pi \tfrac i {2n}))$. Let $m \le 2n$, let $A_m = \{p_1, p_2, \dots, p_m\}$, and let $B \subseteq P_{2n}$ be an arbitary subset of size $m$. Then $\TSP(A_m) \le \TSP(B)$.
\end{lemma}

Before proving Lemma~\ref{lem:circle-subsets}, 
let us argue how it implies our claim. Take any partition of $P_{2n}$ into two sets $A$ and $B$. Without loss of generality, we have $|B| \ge n$. Thus, by Lemma~\ref{lem:circle-subsets}, we have $\TSP(B) \ge \TSP(A_n)$. Therefore, the partition $\{A_n, P_{2n} \setminus A_n\}$ is optimal, as desired. As argued above, we have $\lim_{n \rightarrow \infty} \TSP(A_n) / \TSP(P_{2n}) = \tfrac12 +\tfrac1\pi$. Hence, we have $\gamma(P_{2n},2) \rightarrow \tfrac12 + \tfrac1\pi$, and thus $\gamma(2) \ge \tfrac12 + \tfrac1\pi$. This proves the lower bound of Theorem~\ref{thm:bounds-2}.

To prove Lemma~\ref{lem:circle-subsets}, we need the following technical lemma.

\begin{lemma}\label{lem:sin-add}
		Let $t \in [0, 2\pi]$ and let $a,b \in [0, \tfrac t 2]$ with $a \ge b$. Then $\sin(a) + \sin(t-a) \ge \sin(b) + \sin(t-b)$.
\end{lemma}
\begin{proof}
	Let $f(x) = \sin(x) + \sin(t-x)$. It suffices to show that $f(x)$ is increasing on the interval $[0,\pi]$.
	
	Note that $f'(x) = \cos(x) - \cos(t-x)$. We want to show that $f'(x) \ge 0$ for $x \in [0,\pi]$, i.e.:
	\begin{align}
		\cos(x) \ge \cos(t-x).\label{eq:sin-add}
	\end{align}
	
	Indeed, if $t \le \pi$, then (\ref{eq:sin-add}) is equivalent to $x \le t-x$, which is equivalent to the assumption $x \le \tfrac t 2$. Otherwise, if $t \ge \pi$, then (\ref{eq:sin-add}) is equivalent to $x \le 2\pi - (t-x)$, which is true by assumption $t \le 2\pi$.
\end{proof}

We now prove Lemma~\ref{lem:circle-subsets}. First, recall that an optimal tour does not intersect itself (unless all points collinear). 
Thus, for every subset $B \subseteq P_{2n}$, the optimal tour on $B$ is a simple polygon whose vertices follow the order of the points on the circle.

Now take some subset $B \subseteq P_{2n}$ of size $m$. Without loss of generality, let $p_1 \in B$. We show that $B$ can be transformed into $A_m$ with a series of steps, each of which does not increase $\TSP(B)$.

\begin{lemma}\label{lem:circle-partition-helper}
	Suppose $p_i \in B$, $p_{i+1}, \dots, p_{j-1} \notin B$, and $p_j \in B$. If $B' = B \setminus \{p_j\} \cup \{p_{i+1}\}$, then $\TSP(B') \le \TSP(B)$.
\end{lemma}
\begin{proof}
	Let $\ell > j$ be minimal such that $p_\ell \in B$. (All index operations in this proof are taken $\bmod$ $2n$.)
	By our observation above, the optimal tour of $B$ and the optimal tour of $B'$ only differ in the edges $p_ip_j$, $p_jp_\ell$ (for $B$) and the edges $p_ip_{i+1}$, $p_{i+1} p_\ell$ (for $B'$). See \cref{fig:circle-tour-step} for an illustration.
	
	Observe that $|p_i p_j| = 2 \sin( \pi \tfrac{j-i}{2n})$. Lemma~\ref{lem:sin-add} implies that $ \sin( \pi \tfrac {j-i} {2n} ) + \sin( \pi \tfrac {\ell-j} {2n} ) \ge \sin( \pi \tfrac 1 {2n} ) + \sin( \pi \tfrac {\ell-i-1} {2n} )$, which concludes the proof.
\end{proof}

\begin{figure}
	\centering
	\includegraphics{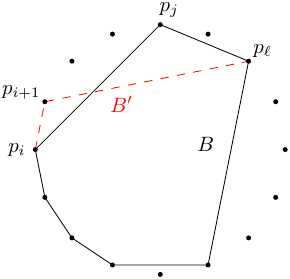}
	\caption{One step in the transformation of a tour $B$ into $A_m$.}\label{fig:circle-tour-step}
\end{figure}

Using Lemma~\ref{lem:circle-subsets}, it is not hard to extend the lower bound to larger $k$:

\restateLBGeneral*
\begin{proof}
	Consider the point set $P_{kn}$. Lemma \ref{lem:circle-subsets} naturally extends to this point set, so we have that a balanced $k$-partition of $P_{kn}$ must consist of $k$ sets of $n$ consecutive points on the circle. Let $A$ denote one of these sets. We have
	\begin{align*}
		& \TSP(P_{kn}) = kn \cdot 2 \sin(\tfrac \pi {kn}), \text{ and}\\
		& \TSP(A) = (n-1) \cdot 2 \sin(\tfrac \pi {kn}) + 2 \sin(\tfrac{\pi(n-1)}{kn}). 
	\end{align*}
	Thus,
	\begin{align*}
		& \gamma(P_{kn})  = \frac{\TSP(A)}{\TSP(P_{kn})} = \frac{n-1}{kn} + \frac{\sin(\tfrac{\pi(n-1)}{kn})}{kn \sin(\tfrac \pi {kn})} \xrightarrow{n \rightarrow \infty} \frac 1 k + \frac{\sin(\tfrac \pi k)}{\pi}.\tag*{\qedhere}
	\end{align*}
\end{proof}

\section{Upper bounds via short diagonals}\label{sec:ub}

We now prove several upper bounds for $\gamma(k)$. Recall that to show $\gamma(k) \le \alpha$, we need to argue that for \emph{every} point set $P$, there is a partition into $k$ subsets $Q_1, \dots, Q_k$ with $\TSP(Q_i) \le \alpha \cdot \TSP(P)$ for all $i \in [k]$.


Our  upper bound for $\gamma(2)$ adapts the idea for the circular point set from \cref{sec:circle}. Essentially, we take the optimal tour $C$ of $P$, split it in two at some point $p$ and its antipodal point $q$, and split the point set accordingly into $Q_1, Q_2$. Let $C_1 = C(p,q) \cup pq$ and $C_2 = C(q,p) \cup pq$. Observe that $C_1$, resp.\ $C_2$ are tours of $Q_1$, resp.\ $Q_2$, and $|C_1| = |C_2| = |C|/2 + |pq|$.
It turns out that we can always find an antipodal pair $p, q$ on $C$ such that the diagonal $|pq|$ is short.

\begin{lemma}\label{lem:short-halving-diag}
	Let $C$ be a closed polygonal curve. Then, there exists a diagonal $pq$ of $C$ such that $|C(p,q)| = |C|/2$ and $|pq| \le \tfrac 1 \pi |C|$.
\end{lemma}
\begin{proof}[Proof sketch.]
	Let $H$ be the convex hull of $C$. Then the width of $H$ is at most $\tfrac1\pi |H| \le \tfrac1\pi |C|$, by Lemma~\ref{lem:cauchy}. This means that there exists a direction $\theta$ such that the width of $H$ in direction $\theta$ is at most $\tfrac1\pi |C|$. Clearly, the same is true for $C$, so any diagonal of $C$ that is parallel to $u_\theta$ has length at most $\tfrac1\pi |C|$. Finally, for \emph{every} direction $\theta$, there exists an antipodal pair $p, q$ such that the diagonal $pq$ is parallel to $u_\theta$. (Essentially, we can rotate $p$ around the curve, obtaining every possible direction.)
	For a full proof, see Lemma~\ref{lem:short-diag}.
\end{proof}

Using Lemma~\ref{lem:short-halving-diag}, the above discussion yields a partition of $P$ into two sets $Q_1$, $Q_2$ with $\TSP(Q_1), \TSP(Q_2) \le \tfrac12 |C| + \tfrac1\pi |C|$. Since $\TSP(P) = |C|$ by assumption, we have $\gamma(2) \le \tfrac12 + \tfrac1\pi$. This concludes the proof of Theorem~\ref{thm:bounds-2}.

Our second upper bound is a simple reduction for non-prime $k$.

\restateUBMultiply*
\begin{proof}
	Let $P$ be a point set.
	By definition, there is a partition of $P$ into sets $Q_1, Q_2, \dots, Q_a$ such that $\tsp(Q_i) \le \gamma(a) \cdot \tsp(P)$ for each $i \in [a]$. We can now further split each $Q_i$ into $b$ sets $Q_{i,1}, Q_{i,2}, \dots, Q_{i,b}$, for a total of $a \cdot b$ sets with $\tsp(Q_{i,j}) \le \gamma(b) \cdot \tsp(Q_i) \le \gamma(b) \cdot \gamma(a) \cdot \tsp(P)$. 
\end{proof}

Our third upper bound combines the short-diagonal technique with the recursive approach of Theorem~\ref{thm:ub-multiply}. The key is a generalization of Lemma~\ref{lem:short-halving-diag}. Essentially, we can always find a short diagonal to split the curve into not only two halves, but into two parts of \emph{any chosen length}. As a first step, we show that for every prescribed curve length and direction, we can find an appropriate diagonal.

\begin{lemma}\label{lem:curve-vec}
	Let $c \colon [0,1] \rightarrow \R^2$ be the parametrization of a closed curve, let $x \in \R$, $0 < x < 1$, and let $u$ be a vector. Then there exists some $t \in [0,1]$ such that the vector $c(t+x) - c(t)$ is orthogonal to $u$.
\end{lemma}
\begin{figure}
	\centering
	\includegraphics{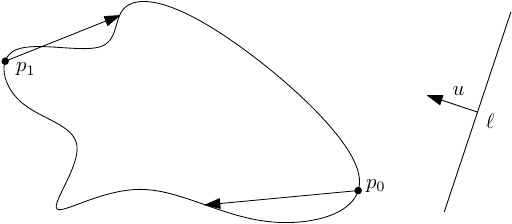}
	\caption{Illustration of Lemma~\ref{lem:curve-vec}.}\label{fig:curve-vec}
\end{figure}
\begin{proof}
	Define $v(t) = c(t+x) - c(t)$ and $f(t) = v(t) \cdot u$. We need to show that there is some $t \in [0,1]$ such that $f(t) = 0$.
	
	Let $p_0 = c(t_0)$ be a point on $c$ minimizing $p_0 \cdot u$. See \cref{fig:curve-vec}. Observe that $p_0$ exists, since the image of $c$ is compact. We claim that $f(t_0) \ge 0$. Indeed, if $f(t_0) < 0$, then, by definition, we have
	\begin{align*}
		( c(t_0+x) - c(t_0) ) \cdot u < 0 \iff c(t_0+x) \cdot u < c(t_0) \cdot u,
	\end{align*}
	which contradicts the minimality of $c(t_0) \cdot u$. 
	
	Similarly, let $p_1 = c(t_1)$ minimize $p_1 \cdot u$. By symmetry, we have $f(t_1) \le 0$.
	Finally, since $f$ is continuous, there must be some $t$ with $f(t) = 0$ by the intermediate value theorem. This concludes the proof.
\end{proof}

We now show how to split curves into parts of a prescribed length.

\begin{restatable}{lemma}{restateShortDiag}\label{lem:short-diag}
	Let $C$ be a closed polygonal curve. Then, for each given $x \in \R$, $0 < x < |C|$, there exists a diagonal $pq$ of $C$ such that $|C(p,q)| = x$ and the length of $pq$ is at most $\tfrac 1 \pi |C|$.
\end{restatable}
\begin{proof}
	Let $c$ be a canonical parametrization of $c$ and let $H$ be the convex hull of $C$. 
	Let $\theta$ be the angle that minimizes the directional width of $H$, i.e., we have $w(H,\theta) = w(H)$. By Lemma~\ref{lem:cauchy}, we have $w(H) \le \tfrac 1 \pi |H| \le \tfrac 1 \pi |C|$.
	
	Now let $u$ be a vector orthogonal to $u_\theta$. By Lemma~\ref{lem:curve-vec}, there exists a $t \in [0,|C|]$ such that $c(t+x) - c(t)$ is orthogonal to $u$, and thus parallel to $u_\theta$. Let $p = c(t)$ and $q = c(t+x)$. Then $|C(p,q)| = x$ by definition, and the length of $pq$ is bounded by $w(C, \theta) \le w(H,\theta) \le \tfrac 1 \pi |C|$. This concludes the proof.
\end{proof}

Using Lemma~\ref{lem:short-diag}, we obtain the following generic upper bound:

\begin{lemma}\label{lem:ub-diag-helper}
	Let $k, a, b \in \N$ with $k = a + b$, and let $x \in \R$, $0 < x < 1$. Then $\gamma(k) \le \max{\left\{ (x+ \tfrac 1 \pi) \cdot \gamma(a), (1-x + \tfrac 1 \pi) \cdot \gamma(b) \right\}}$.
\end{lemma}
\begin{proof}
	Let $P$ be a point set, and let $T$ be a tour of $P$. Without loss of generality, we have $|T| = 1$.
	
	We first use Lemma~\ref{lem:short-diag} to obtain a diagonal $pq$ of length at most $\tfrac 1 \pi$ such that $|C(p,q)| = x$. As in the proof of the upper bound of Theorem~\ref{thm:bounds-2}, we take the two tours $C_1 = C(p,q) \cup pq$ and $C_2 = C(q,p) \cup pq$, with $|C_1| \le x + \tfrac 1 \pi$ and $|C_2| \le 1 - x + \tfrac 1 \pi$.
	
	Now partition $P$ into $P_1, P_2$ such that $P_1 \subseteq C_1$ and $P_2 \subseteq C_2$. By definition, we have $\TSP(P_1) \le x + \tfrac 1 \pi$ and $\TSP(P_2) \le 1 - x + \tfrac 1 \pi$. Optimally partitioning $P_1$ into $a$ parts will yield tours of length at most $\gamma(a) \cdot \TSP(P_1)$, by definition, and similarly, optimally partitioning $P_2$ into $b$ parts will yield tours of length $\gamma(b) \cdot \TSP(P_2)$. Overall, we have $a+b = k$ tours of length at most $\max{\left\{ \gamma(a) \cdot \TSP(P_1), \gamma(b) \cdot \TSP(P_2)\right\}}$, as desired.
\end{proof}

Optimizing for $x$ in Lemma~\ref{lem:ub-diag-helper} yields
\[ x = \frac{\gamma(b)}{\gamma(a) + \gamma(b)} + \frac{\gamma(b) - \gamma(a)}{\pi \cdot (\gamma(a) + \gamma(b))}. \]
Finally, with a simple calculation, we have:

\restateSplitUB*

Note that the upper bound of Theorem~\ref{thm:bounds-2} follows as a special case when $a = b = 1$ (recall that $\gamma(1) = 1$).

\newpage

\bibliography{references.bib} 
\bibliographystyle{alphaurl}

\end{document}